\documentclass[journal,twoside]{IEEEtran}

\newcommand{\Prob}{\mathbb{P}}

\usepackage{amsmath,amsfonts}
\usepackage{amsthm}%
\usepackage{tikz}%
\usepackage{mathrsfs}
\usepackage[caption=false,font=normalsize,labelfont=sf,textfont=sf]{subfig}
\usepackage{textcomp}
\usepackage{stfloats}
\usepackage{url}
\usepackage{verbatim}
\usepackage{graphicx}
\usepackage{cite}
\usepackage{hyperref}%
\usepackage{cleveref}%

\begin{document}

\newtheorem{theorem}{Theorem}[section]
\newtheorem{quest}{Question}
\newtheorem{corollary}{Corollary}[section]

\theoremstyle{remark}
\newtheorem{remark}{Remark}[section]
\newtheorem{example}{Example}[section]

\title{Equivalence of Serial and Parallel A-Posteriori Probability Decoding in Digital Audio Broadcasting}
\author{Andrea Di Giusto, Wim van Houtum~\IEEEmembership{Senior Member,~IEEE}%
, Alberto Ravagnani,  Yan Wu
\thanks{A. Di Giusto is with NXP Semiconductors and Eindhoven University of Technology (email: a.di.giusto@tue.nl).}%
\thanks{W. J. van Houtum is with NXP Semiconductors and Eindhoven University of Technology (email: wim.van.houtum@nxp.com).}%
\thanks{A. Ravagnani is with Eindhoven University of Technology (email: a.ravagnani@tue.nl).}%
\thanks{Y. Wu is with NXP Semiconductors (email: yan.wu\textunderscore 2@nxp.com).}%
\thanks{A. Di Giusto and A. Ravagnani are supported by the European Commission through grant 101072316.}%
  }
\markboth{}%
{Di Giusto \MakeLowercase{\textit{et al.}}: Equivalence of Serial and Parallel APP Decoding in DAB Systems}

\maketitle
\begin{abstract}
Motivated by applications to digital audio broadcasting (DAB) systems, we study the a-posteriori probabilities (APPs) of the coded and information bits of the serial concatenation
  of multiple convolutional codewords.
The main result of this paper is a proof that the APPs of the input bits do not change when considering the concatenation of multiple codewords as a received sequence.
This is a general theoretical result, which remains valid for every convolutional code, as long as the encoder goes back to the zero state at the end of each codeword.
An equivalent approach for serial concatenation in Viterbi decoding is described.
The applicability of our result to DAB systems, where interleaving and modulation are accounted for, is investigated through Matlab simulations.
We show that the Bit Error Rate (BER) of the simulated DAB system does not change when decoding multiple transmitted codewords as one serially concatenated sequence, 
  even when considering all the features of a DAB system.
\end{abstract}

\vspace{-0.5cm}
\section{Introduction}
\IEEEPARstart{D}{AB} systems often carry multiple convolutionally encoded data components (services) simultaneously.
For example, the standard specifications of~\cite{etsi300} include a Bit Interleaved Coded Modulation (BICM) part and an Orthogonal Frequency Division Multiplexing (OFDM) part,
    where each subcarrier is modulated by~$\pi/4$ differentially encoded quaternary phase shift keying (DE-QPSK).
In the BICM, multiple data streams are encoded separately, which are then aggregated in the Common Interleaved Frames (CIFs) in the Main Service Channel (MSC).
A CIF is then composed by multiple \textit{parallel} codewords and, at the receiver's end, it is usual to decode each of these codewords separately.
Iterative decoding techniques~\cite{hagenauer1996iterative} have been proposed to enhance the performance of DE-QPSK streams~\cite{peleg2000iterative}, and further analyzed in the context of OFDM systems 
    in~\cite{van2012two}.
In this framework, the DAB system is a concatenated code, with a convolutional outer code and the DE-QPSK modulation acting as an inner code.
To each parallel codeword correspond only few OFDM symbols; since the parallel received words are decoded independently, this affects the performance of the iterative part~\cite{van2012two}.
In this paper, we propose and study a solution to this issue: to decode the different parallel convolutional codewords as one long 
    \textit{serial} codeword, by concatenating them one after the other.
As a first step, we isolate the coding problem: we consider the APPs of the information and coded bits, and we prove that 
    they are unchanged if we consider the serial received word instead of the parallel words. 
The only assumptions required for this result are (1) that each parallel encoder goes back to the zero state at the end of the encoding, and (2) that the decoder knows
    the length of each codeword.
These assumptions are not restrictive, and in particular they are met by the specifications of the DAB standard~\cite{etsi300}; in practice, the APPs are computed using the 
    BCJR algorithm~\cite{bahl1974optimal}.
For receivers using Viterbi decoding~\cite{viterbi1967error}, we briefly describe a modification of the algorithm allowing for serial decoding.
To support the applicability of our ideas to practical DAB systems, we run Matlab simulations comparing the proposed serial decoding to the standard parallel procedure.
We compute the Log-Likelihood Ratios (LLR) of the information and coded bits, and find that they are unchanged in the serial decoding.
This aligns with our expectations, since the LLRs are computed directly from the APPs, and supports the applicability of serial decoding of multiple words to practical DAB systems.
We compute and show BER curves of the simulated systems.
The remainder of this work is organised as follows: in \Cref{sec:notation_problem} we introduce the necessary notation and formalize our coding problem, 
    in \Cref{sec:decoding_problem} we prove our main result on the serial decoding of multiple words, and in \Cref{sec:simulation} we comment on our simulation results.
\Cref{sec:conclusions} concludes the paper.

\section{Notation and problem statement}\label{sec:notation_problem}
\paragraph{Notation} let~$X=x_1\ldots x_n$ be a sequence of~$n$ elements of some set, and~$1\leq i\leq j\leq n$, then~$[X]_i^j=x_i\ldots x_j$ denotes the 
    subsequence of~$X$ from the~$i^{th}$ element to the~$j^{th}$, including the extremes.
If~$X$ and~$Y$ are sequences, then~$XY$ denotes their concatenation; hence an example of our notation is~$X=[X]_1^i[X]_{i+1}^n=x_1\ldots x_i x_{i+1}\ldots x_n$. 
The convolutional encoding of a sequence~$X$ with the mother code of~\cite{etsi300} (rate~$1/4$, memory~$6$, generated by the polynomials~$[133,171,145,133]_8$) is denoted by~$Y=convenc(X)$.
In trellis operations, time indexes states (starting from~$0$) and trellis sections index transitions 
  (starting from~$1$).
Hence trellis section~$1$ ($t$) refers to the transition between states at time~$0$ ($t-1$) and~$1$ ($t$).
~$\Prob(\cdot)$ denotes probability.
\paragraph{Problem statement}
we consider the CIF structure of~\cite[Section 13]{etsi300}, where one CIF is composed of~$N$ parallel subchannels.
For simplicity, we assume to each subchannel corresponds one \textit{parallel} convolutional codeword.
To improve the performance of the iterative decoding, we use the whole CIF as one codeword; we will call this the \textit{serial} approach.
Abstracting from practical DAB specifications, we consider the following coding setup.
For every~$i=1,\ldots,N$ let~$B_i$ be bit sequences of arbitrary length.
To encode these sequences with the rate 1/4 mother code of~\cite{etsi300}, six zero bits, denoted by~$0_6$, are appended at the end of each~$B_i$, yielding the~$N$ parallel 
    information sequences~$B_i0_6$ of length~$K_i=6+\textnormal{length of }B_i$.
For bookkeeping reasons, we let~$L_i=\sum_{i=1}^NK_i$ and~$L_0=0$.
Let~$X_i=convenc(B_i0_6)$ be the parallel codewords obtained encoding the information sequences, and let~$B$ and~$X$ denote the serial concatenations
\begin{equation*}
    B=B_10_6B_20_6\ldots B_N0_6\quad\text{and}\quad X=X_1X_2\ldots X_N\;.
\end{equation*}
It is convenient to think of the codewords~$X_i$ as sequences formed by 4-bit symbols, because as such they have length~$K_i$ and each bit in an information sequence 
    corresponds to a 4-bit symbol in the encoded sequence.
With this correspondence in mind, both~$B$ and~$X$ have length~$L_N$, and using our notation for sequences we have~$B_i=[B]_{L_{i-1}+1}^{L_i}$ and~$X_i=[X]_{L_{i-1}+1}^{L_i}$.
The single bits of~$B$ are denoted by~$b_t$, meaning~$B=b_1\ldots b_{L_N}$.
Suppose serial codeword~$X$ is sent over a channel, and an altered version~$Y$ of~$X$ is received and has to be decoded.
Depending on the decoder in use,~$Y$ can be either a sequence of hard or soft bit-metrics; as done with~$X$, we consider it as a sequence of 4-bit-metrics symbols,
    implying~$Y$ too has length~$L_N$.
Each 4 bit-metrics correspond to 4 encoded bits, which in turn correspond to 1 information bit~$b_t$ in~$B$.
Facing the task of recovering the information sequences~$B_i0_6$ from~$Y$, the parallel approach is to split up~$Y$ in~$N$ received parallel sequences~$Y_i=[Y]_{L_{i-1}+1}^{L_i}$, 
    each of which is then decoded independently.
This is a reasonable approach if the decoding procedure is only performed once; anyway, when performing iterative decoding in DAB systems, 
    information is exchanged back and forth between the outer codeword (the convolutionally encoded CIF) and the inner codeword (the DE-QPSK modulated subcarrier in OFDM).
Each subchannel in the CIF only corresponds to few symbols in the OFDM system, reducing the effectiveness of the information exchange in the iterations between inner and outer code.
For this reason, we investigate the following
\begin{quest}\label{question}
    Is it possible to decode~$Y$ without dividing it and still get the same information about the information sequences/performance of the coding scheme?
    How does this serial approach of the decoding of~$Y$ compare to the parallel approach?
\end{quest}
We will first study this problem from a theoretical standpoint, and later from a practical angle.
To familiarize with the notations and the comparison between serial and parallel approach, we consider the instructive example of encoding.
\begin{example}\label{ex:encoding}
    To define the serial words we followed an encode-then-concatenate rationale (parallel encoding followed by concatenation).
    What if we tried to encode directly the serial information word~$B$?
    This question is the encoding analogue of Question \ref{question}, and it has simple but instructive answer.
    To compare serial and parallel encoding, let~$\tilde{X}=convenc(B)$.
    Since the convolutional code has memory, in general it is not true that encoding and serial concatenation of information sequences commute.
    Recall that in the notation above we have~$X_i=[X]_{L_{i-1}+1}^{L_i}$.
    It is clear that~$[\tilde{X}]_1^{L_1}=[X]_1^{L_1}$, because both encoders start from the 0 state and they are fed the same input bits, that are the bits of~$B_10_6$.
    Because of the six zeros at the end of the information word~$B_10_6$, the trellis state in the encoding process of~$B$ at trellis section~$L_1$ is the zero state.
    Hence at trellis section~$L_1+1$ (that is, when the first bit of the second information word is the input) of the encoding of~$B$, the initial state of the transition 
        is the 0 state, just as it is in the beginning of the parallel encoding of~$B_20_6$.
    It follows that the outputs of these two steps must coincide, meaning~$[\tilde{X}]_{L_1+1}^{L_1+1}=[X_2]_1^{1}=[X]_{L_1+1}^{L_1+1}$.
    The encoding of~$B$ then proceeds following the same path in the trellis of the encoding of~$B_20_6$, and again the six zeros at the end bring the state 
        of the encoder back to 0 at time~$L_2=K_1+K_2$.
    It is then clear that our reasoning applies to all the~$N$ parallel words: at time~$L_i$, the state of the encoder that is encoding~$B$ will always be zero.
    Consequently, we have~$X=\tilde{X}$, meaning serial and parallel encoding are equivalent, and encoding and serial concatenation commute.
    As a final observation, notice that the six tail bits at the end of each information sequence~$B_i0_6$ play a key role in this reasoning; 
        this will be important also in the study of Question \ref{question}.
\end{example}
In the rest of the paper, we extend the comparison between serial and parallel framework to the computation of the APPs of the information and encoded bits, 
    and then to the DAB framework.

    \section{Convolutional decoding}\label{sec:decoding_problem}
\noindent As for the encoding, in principle there is no reason why the parallel and serial procedures should give the same result.
Anyway, some extra information is known to the decoder: the last six bits of each information word~$B_i0_6$ are always zero.
The position of the last six zero bits in an information word can be inferred from the length of each codeword.
As we have seen in \Cref{ex:encoding}, the presence of the tail bits makes the parallel and serial encoding of multiple information words equivalent;
    can it have a similar effect on the computation of the APPs of the information/coded bits?
Following~\cite{bahl1974optimal}, the encoder is a Markov chain~$S_t$ with known initial state~$S_0=0$; in our case, the chain has 64 states, as the CC of~\cite{etsi300} has memory 6 and~$2^6=64$.
The transitions of the Markov source are regulated by the information bits, and produce the encoded bits.
Since our code has rate 1/4, for each transition~$S_{t-1}\rightarrow S_t$ we have one input bit~$b_t$ and four output bits~$Z_t=z_t^1z_t^2z_t^3z_t^4$.
In this notation, by \Cref{ex:encoding} we have that~$X=X_1\ldots X_N=Z_1\ldots Z_{L_N}$ is the encoding of~$B=B_1\ldots B_N=b_1\ldots b_{L_N}$. 
To determine the APPs, we study the process~$S_t$ and its transitions, conditioned on the sequence~$Y$ (serial) or the sequences~$Y_i$ (parallel).
For every~$t=1,\ldots,L_N$, let~$i(t)$ be the index~$i=1,\ldots,N$ such that~$X_i$ is the parallel word encoding~$b_t$.
In other words,~$i(t)$ is the unique~$i=1,\ldots,N$ s.t.~$t\in(L_{i-1},L_i]$; moreover we denote by~$L^-(t),L^+(t)$ the extremal points of the \textit{closed} interval containing~$t$,
    meaning~$L^-(t)=L_{i(t)-1}+1,L^+(t)=L_{i(t)}$.
\begin{figure}[ht] 
    \centering
    \includegraphics[width=8cm,scale=0.5]{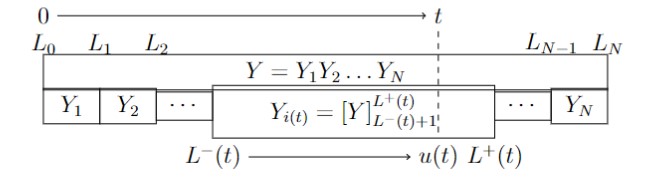}
    \caption{Illustration for the notation.}
    \label{notation_BCJR} 
    \vspace{-0.4cm}
\end{figure}
Hence for every~$t$,~$Y_{i(t)}=[Y]_{L^-(t)}^{L^+(t)}$ (recall~$Y$ has length~$L_N$ as a sequence of 4-bit-metrics), which will be our notation for the parallel decoders'inputs; 
    see \Cref{notation_BCJR} for a graphical description of the notation.
This notation is convenient to compare the parallel and serial decoding, as in the parallel APP decoding we guess based on~$Y_{i(t)}$,
    while in the serial APP decoding we guess based on the whole~$Y$.

\paragraph{Information bit distribution} in terms of information bits, we have the following problem: we want to compute the probability that each bit~$b_t$ is either 
    0 or 1, for all~$t=1,\ldots,L_N$, conditioned either on the corresponding received parallel word~$Y_{i(t)}$ or on the received full serial word~$Y$.
The following is the key result of this section.
\begin{theorem}\label{thm:APP_serial_parallel}
    Under the assumption that the positions of the tail bits are known to the receiver, the parallel and serial APPs of the information bits coincide.
    Formally, for all~$t=1,\ldots,L_N$ we have
    \begin{equation}\label{eq:thm_APP_statement}
        \Prob(b_t=0\mid Y)=\Prob(b_t=0\mid Y_{i(t)}).
    \end{equation}
\end{theorem}
\begin{proof}
    The theorem follows from interpreting the knowledge about the tail bits in terms of states~$S_t$ of the encoder and then using the Markov property.
    By~\Cref{ex:encoding}, encoding~$B$ is equivalent to parallel encoding followed by concatenation.
    For every~$i=1,\ldots,N$, the tail bits at the end of each subsequence~$[B]_{L_{i-1}+1}^{L_i}=B_i0_6$ bring the encoder~$S_t$ back to the zero state. 
    The receiver knows the positions of the tail bits in~$B$, and since the~$i^{th}$ parallel word ends at time~$L_i$, 
        it follows that we have~$S_{L_i}=0$ for all~$i=0,\ldots,N$ (as~$S_{L_0}=S_0=0$ by definition).\\
    Following~\cite{bahl1974optimal}, let~$\mathscr{A}_{t}=\{m=(s_{t}^{(1)}\ldots s_{t}^{(6)})\;:\;s_{t}^{(1)}=0\}$ 
        be the set of states corresponding to a~$0$ input at time~$t$, then we have
    \begin{equation}\label{eq:serial_APP}
        \Prob(b^t=0\mid Y)=\sum_{m\in \mathscr{A}_t}\frac{\Prob(S_t=m,Y)}{\Prob(Y)}=\frac{1}{\Prob(Y)}\sum_{m\in \mathscr{A}_t}\lambda_t(m)
    \end{equation}
    and similarly
    \begin{equation}
        \Prob(b^t=0\mid Y_{i(t)})=\frac{1}{\Prob(Y_{i(t)})}\sum_{m\in \mathscr{A}_t}\lambda^P_t(m)\;.
    \end{equation}
    where ~$\lambda_t(m)=\Prob(S_t=m,Y)$ and ~$\lambda^P_t(m)=\Prob(S_t=m,Y_{i(t)})$ are the parallel and serial joint probabilities respectively.
    Expliciting the receiver's knowledge of the states~$S_{L_i}=0$ for~$i=0,\ldots,N$ (and especially for~$i=i(t)-1$,~$i=i(t)$) we have 
       ~$\Prob(S_t=m\mid Y)=\Prob(S_t=m\mid Y,S_{L^-(t)}=0,S_{L^+(t)}=0)$.
    Since~$S_t$ is a Markov process, its state at time~$t$ is independent from any observation prior to (after) the known state~$S_{L^-(t)}=0$ ($S_{L^+(t)}=0$).
    It follows that
    \begin{align}
        \Prob(&S_t=m\mid Y,S_{L^-(t)}=0,S_{L^+(t)}=0)\nonumber\\
        &=\Prob(S_t=m\mid [Y]_{L^-(t)+1}^{L_N},S_{L^-(t)}=0,S_{L^+(t)}=0)\nonumber\\
        &=\Prob(S_t=m\mid [Y]_{L^-(t)+1}^{L^+(t)},S_{L^-(t)}=0,S_{L^+(t)}=0)\;,
    \end{align}
    where by definition~$\Prob(S_t=m\mid Y_{i(t)},S_{L^-(t)}=0,S_{L^+(t)}=0)=\Prob(S_t=m\mid Y_{i(t)})$.
    Then we have~$\Prob(S_t=m\mid Y)=\Prob(S_t=m\mid Y_{i(t)})$, and hence
    \begin{align}
        \lambda_t(m)&=\Prob(Y)\Prob(S_t=m\mid Y)=\Prob(Y)\Prob(S_t=m\mid Y_{i(t)})\nonumber\\
        &=\frac{\Prob(Y)}{\Prob(Y_i)}\lambda^P_t(m)\;.
    \end{align}
    Replacing this in \Cref{eq:serial_APP} and simplifying, we get 
    \begin{equation}
        \Prob(b^t=0\mid Y)=\frac{1}{\Prob(Y)}\sum_{m\in \mathscr{A}_t}\frac{\Prob(Y)}{\Prob(Y_i)}\lambda^P_t(m)=\Prob(b^t=0\mid Y_{i(t)})
    \end{equation}
    which is precisely what we wanted to prove.
\end{proof}

\begin{remark}
    Despite being framed for the specific DAB convolutional code, this result is actually independent from the choice of the code and/or any other operations 
        applied to the codewords, such as puncturing or interleaving. 
    As long as the serial received word is formed by concatenating the parallel received words, and the termination conditions are met in the encoding process,
    \Cref{thm:APP_serial_parallel} holds.
    For example, in DAB system, the encoded words~$X_i$ are often punctured to obtain higher rates and unequal error protection; 
        see~\cite[Section 11]{etsi300} for many examples of different protection profiles.
    If we want to include puncturing in our computations, we simply denote by~$\bar{X}_i$ the punctured parallel codewords and by~$\bar{X}$ again their concatenation.
    Notice that~$\bar{X}=P(convenc(B))$ for a suitable puncturing pattern~$P$ (the "concatenation" of the puncturing patterns used for the~$X_i$).
    Denoting as~$\bar{Y}$ the sequence received over the AWGN channel, after the required de-puncturing we get again a sequence~$Y$ 
        of length~$L_N$ (in 4-symbols) that can be fed to a decoder.
    The case of interleaving is even more simple, as interleaving is just a bijection that scrambles the sequences but does not alter the information contained in them.
    If a code with different constraint length is used, the number of tail bits at the end of each information word needs to be adapted accordingly.
\end{remark}

The BCJR algorithm outputs a guess for a bit~$b_t$ based on its distribution conditioned on the observed sequence.
A simple consequence of the theorem we just proved is that the parallel and serial guess for~$b_t$ coincide.

\begin{corollary}
Assume that the positions of the tail bits are known, and the BCJR algorithm is used as a decoder.
Let~$\hat{b}_t$ and~$\hat{b}^{i(t)}_t$ be the serial and parallel outputs of the decoders corresponding to bit~$b_t$ in the information sequence, then we have
       ~$\hat{b}_t=\hat{b}^{i(t)}_t$.
\end{corollary}
\begin{proof}
    Using the BCJR algorithm~\cite{bahl1974optimal}, the parallel (serial) guess of~$b_t$ is~$\hat{b}^{i(t)}_t=0$ ($\hat{b}_t=0$) if~$\Prob(b_t=0\mid Y_{i(t)})\geq0.5$
    (if~$\Prob(b_t=0\mid Y)\geq0.5$), and~$\hat{b}^{i(t)}_t=1$ ($\hat{b}_t=1$) otherwise.
    Since by \Cref{thm:APP_serial_parallel} we have~$\Prob(b_t=0\mid Y)=\Prob(b_t=0\mid Y_{i(t)})$, the thesis follows.
\end{proof}

\paragraph{Encoder output distribution} when decoding iteratively, we are interested also in the APPs of the 4 output bits~$Z_t=z_t^1z_t^2z_t^3z_t^4$ of the encoder at time~$t$,
    as they constitute the extrinsic information used as an a priori input for the next iteration~\cite{hagenauer1996iterative}.
Since the output digits depend directly from the encoder process~$S_t$, the reasoning we made about the APP distribution of the input bit translates also in this setting. 
The considerations based on the Markovianity of~$S_t$ and the known zero states still hold. 
For~$j=1,\ldots,4$ we let~$\mathscr{B}_t^{j}$ be the set of transitions~$m'\rightarrow m$ s.t. the~$j^{th}$ bit of the corresponding output sequence~$Z_t$ is 0; then 
\begin{align}
    \Prob(z_t^{j}=0\mid Y)&=\sum_{m\in \mathscr{B}_t^{(j)}}\frac{\Prob(S_{t-1}=m',S_t=m,Y)}{\Prob(Y)}\nonumber\\
    &=\frac{1}{\Prob(Y)}\sum_{m\in \mathscr{B}_t^{(j)}}\sigma_t(m)
\end{align}
and
\begin{align}
    \Prob(z_t^{j}=0\mid Y_{i(t)})&=\sum_{m\in \mathscr{B}_t^{(j)}}\frac{\Prob(S_{t-1}=m',S_t=m,Y_{i(t)})}{\Prob(Y_{i(t)})}\nonumber\\
    &=\frac{1}{\Prob(Y_{i(t)})}\sum_{m\in \mathscr{B}_t^{(j)}}\sigma^P_t(m)\;.
\end{align}
where~$\sigma_t(m)$ and~$\sigma^P_t(m)$ are again joint probabilities.
By Markovianity of~$S_t$, and using the (implicit) knowledge of the states~$S_{L_i}=0$ for~$i=0,\ldots,N$, as in the proof of \Cref{thm:APP_serial_parallel} we have
\begin{align}
    \Prob&(S_{t-1}=m',S_t=m\mid Y)=\nonumber\\
    &=\Prob(S_{t-1}=m',S_t=m\mid Y,S_{L_{i(t)-1}}=0,S_{L_i}=0)\nonumber\\
    &=\Prob(S_{t-1}=m',S_t=m\mid Y_{i(t)},S_{L_{i(t)-1}}=0,S_{L_i}=0)\nonumber\\
    &=\Prob(S_{t-1}=m',S_t=m\mid Y_{i(t)})
\end{align}
and hence it follows that
\begin{align}
    \sigma_t(m',m)&=\Prob(Y)\Prob(S_{t-1}=m',S_t=m\mid Y)\nonumber\\
    &=\Prob(Y)\Prob(S_{t-1}=m',S_t=m\mid Y_{i(t)})\nonumber\\
    &=\frac{\Prob(Y)}{\Prob(Y_i)}\sigma^P_t(m',m)\;.
\end{align}
Replacing this in the equation for~$\Prob(z_t^{j}=0\mid Y)$ we get 
\begin{align}
    \Prob(z_t^{j}=0\mid Y)&=\frac{1}{\Prob(Y)}\sum_{m\in \mathscr{B}_t^{(j)}}\frac{\Prob(Y)}{\Prob(Y_i)}\sigma^P_t(m)\nonumber\\
    &=\Prob(z_t^{j}=0\mid Y_{i(t)})\;.
\end{align}
It follows then that the parallel and serial guess for~$z_t^{j}$ coincide: the parallel guess is
\begin{equation*}
    \hat{z}_{t}^{j,i(t)}=\begin{cases}
        0&\textnormal{ if }\Prob(z_t^{j}=0\mid Y_{i(t)})\geq0.5\\
        1&\textnormal{ otherwise}
    \end{cases}
\end{equation*}
and the serial guess~$\hat{z}_t^{j}$ is obtained by replacing~$Y_{i(t)}$ with~$Y$.

\paragraph{Viterbi decoding} 
for efficiency reasons, many DAB receivers use the Viterbi algorithm.
The knowledge about the position of the tail bits can be used to get a serial Viterbi decoding for the parallel codewords.
For~$i=1,\ldots,N$ we let~$V_i=vitdec(Y_i)$ be the output of the parallel Viterbi decoders,~$\pi_i$ is the trellis paths corresponding to~$V_i$, by~$BM_{j,k}^i$ the~$k^{th}$ branch metric on the~$j^{th}$ path, by~$PM^i_{j}$ 
    the path metric of the~$j^{th}$ path.
In the concatenated setting,~$V=vitdec(Y)$,~$\pi$ is the trellis path corresponding to~$V$,~$BM_{j,k}$ is the~$k^{th}$ branch metric of the~$j^{th}$ path,~$PM_j$ is the~$j^{th}$ path metric.
To simplify the tractation, we assume that~$\pi_i$ and~$\pi$ are the are unique paths with the highest metrics.
The information about the last six bits of each information word~$B_i0_6$ can be used during the decoding as follows.
Consider~$V_1=vitdec(Y_1)$ and the surviving paths in the Viterbi algorithm at time~$K_1-6$; we know that the output of the decoder 
    in the subsequent 6 sections should be zero, as the section correspond to tail bits.
For this reason we can twist the branch metrics of the last 6 trellis sections to be
\begin{equation*}
    BM_{j,k}^1=\begin{cases}
        0&\textnormal{if the output corresponds to a decoded 0}\\
        -\infty&\textnormal{otherwise}
    \end{cases}
\end{equation*}
(as~$0=\log(1)$,~$-\infty=\log(0)$) for all paths~$j$ and all branch indexes~$k=K_1-5,\ldots,K_1$.
This has the desired effect in terms of output bits: every path corresponding to a decoded 1 in the last 6 digits will have path metric of~$-\infty$, 
    meaning it will not survive when compared to any path with zeros in the desired positions.
We can also imagine this modification as a funnel in the trellis: suppose~$j$ is a surviving path at time~$K_1-6$, then~$j$ has potentially two extensions to a surviving 
    path at time~$K_1-5$, corresponding to a 0 and a 1 in the decoded sequence.
With the modified branch metrics, we can immediately discard the extension corresponding to an output 1, since it will have path metric~$-\infty$ in the end.
It follows that surviving paths with finite metric at time~$K_1-5,K_1-4,\ldots,K_1$ will be funneled towards the 0 state, which they will all reach at time~$K_1$.
Essentially, modifying the last six branch metrics ensures that~$\pi_1$ is the surviving path passing through the 0 state at time~$K_1$ in the standard Viterbi algorithm without modification.
We apply the modification to all the last 6 branches of the parallel decoding processes: for all~$i=1,\ldots,N$, for all~$k=K_{i}-5,\ldots,K_i$, and for all paths~$j$ we let
$BM_{j,k}^i=0$ if the corresponding output is a zero,~$-\infty$ otherwise.
Notice that this modification implies that in each parallel decoder, the surviving path~$\pi_i$ starts and ends in the zero state.
Hence we can concatenate all of these paths to get a path~$\pi_1\pi_2\ldots\pi_N$ starting and ending in the zero state, 
    and passing through the zero state at all times~$L_i$,~$i=1,\ldots,N$.
The corresponding branch metric modification in the decoding of~$Y$ has to be done on all branches corresponding to trellis sections~$k=L_i-5,L_i-4,\ldots,L_i$, 
    for all~$i=1,\ldots,N$, for every path~$j$:
\begin{equation*}
    BM_{j,k}=\begin{cases}
        0&\textnormal{if the output corresponds to a decoded 0}\\
        -\infty&\textnormal{otherwise}
    \end{cases}\;.
\end{equation*}
As in the parallel decoding, this modification has the desired effect: all paths corresponding to output sequences~$V$ not having the required six-zero 
    blocks in the right positions will have path metric~$-\infty$, and hence will surely be discarded when they meet any path with finite metric.
The corresponding funnel view is that we have~$N$ funnels at branches~$L_i-5,\ldots,L_i$ for all~$i=1,\ldots,N$ in the trellis that is being explored in the serial 
    decoding procedure.
Consequently, the path~$\pi$ corresponding to the output will be a path passing through state~$0$ at times~$L_i$,~$i=1,\ldots,N$.
But then we have that~$\pi$ is the composition of~$N$ paths~$\tilde\pi_i$ starting from state~$0$ at trellis section~$L_{i-1}$ and returning to is at section~$L_i$.
In fact, when we compare all the paths passing through the 0 state at time~$L_1=K_1$, the Viterbi rule leaves us with exactly one path going from the zero 
    state at time 0 to the zero state at time~$L_1$.
Since the inputs of the decoder up to this point are exactly~$[Y]_1^{L_1}=Y_1$, it follows that the surviving path up to this point must be~$\tilde\pi_1=\pi_1$.
This reasoning extends to the part of the trellis between time~$L_1$ and time~$L_2$: the path maximizing the path metric in this part is~$\tilde\pi_2=\pi_2$, 
and the only survivor at time~$L_2$ for the decoding of~$Y$ will be the concatenation~$\pi_1\pi_2$.
This argument applies to the surviving paths at time~$L_i$ for all~$i=1,\ldots,N$.
Hence the surviving path~$\pi$ corresponding to the output~$V=vitdec(Y)$ is actually the concatenation of the surviving paths~$\pi_i$ corresponding to the outputs~$V_i=vitdec(Y_i)$.

\section{Simulation results}\label{sec:simulation}
\noindent We support our findings with two Matlab simulations.
\paragraph{First simulation} in this simulation, we model as closely as possible the coding problem of \Cref{sec:decoding_problem}, disregarding the specifications of DAB: 
    we use QPSK modulation over AWGN channel, and no interleaving.
This simulation features 12 parallel subchannels, each of which encodes an information word (random bit sequence) of fixed length and is puncutred according to~\cite[Table 29, page 131]{etsi300}. 
For \Cref{fig:sim1} we used information words of length 4808 bits (including the tail bits), and the 12 puncturing indexes for the subchannels are (20,15,21,24,9,10,8,17,20,21,24,23)
By giving as input to the APP decoder the appropriate a priori LLR vectors that reflect the knowledge of the states~$S_{L_i}$ for~$i=1,\ldots,N$ 
    (see \Cref{thm:APP_serial_parallel}), we find the two BER curves of \Cref{fig:sim1}: the serial and parallel BER are exactly superposed.
This is expected, given that the APPs of the bits determine the outputs of the decoder.
\begin{figure}[ht]
    \centering
    \includegraphics[width=7cm,scale=0.45]{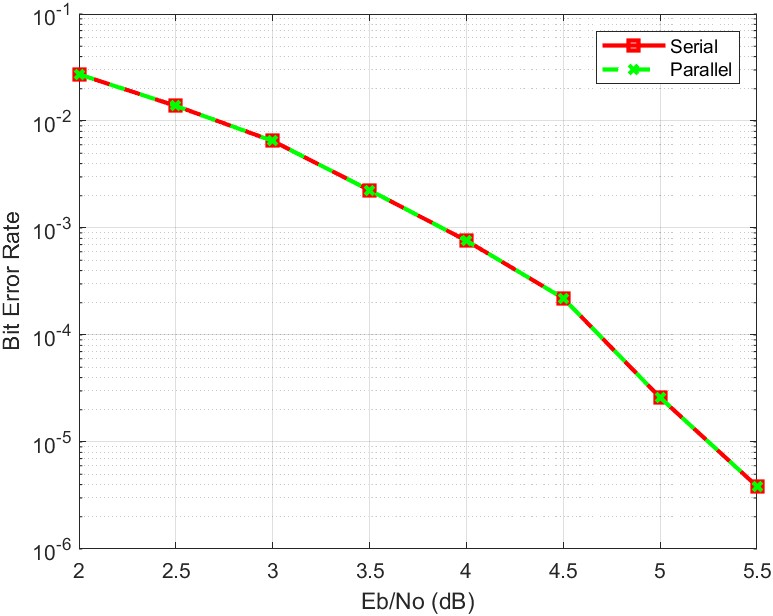}
    \vspace{-0.3cm}
    \caption{First simulation.}
    \label{fig:sim1}
    \vspace{-0.4cm}
\end{figure}
\paragraph{Second simulation} this simulation shows that our results still hold in a DAB-like framework, 
    with DE-QPSK modulation with OFDM, time/frequency interleaving and other features of the standard.
The information about the positions of the tail bits is given to the APP decoder in the form of \textit{a priori} information; in a real-life scenario, 
    it is transmitted in the Fast Information Channel as Multiplex Configuration Information.
We set up a simulation where each CIF is divided into 4 subchannels, each encoded using protection level 3-A with a data rate of 288 kbit/s.
The parameters are specified according to Table 7, p.51, and Tables 33 and 34, p.176 of~\cite{etsi300}; each subchannel has rate 1/2.
We used Transmission Mode I of~\cite{etsi300} over AWGN channel, and coherent demodulation~\cite{van2012two} at the receiver.
In \Cref{fig:sim2} we show the BER curves for 4 parallel subchannels, their aggregate BER and the BER of their serial concatenation; as expected, the latter two coincide.
\begin{figure}[ht]
    \centering
    \includegraphics[width=7cm,scale=0.45]{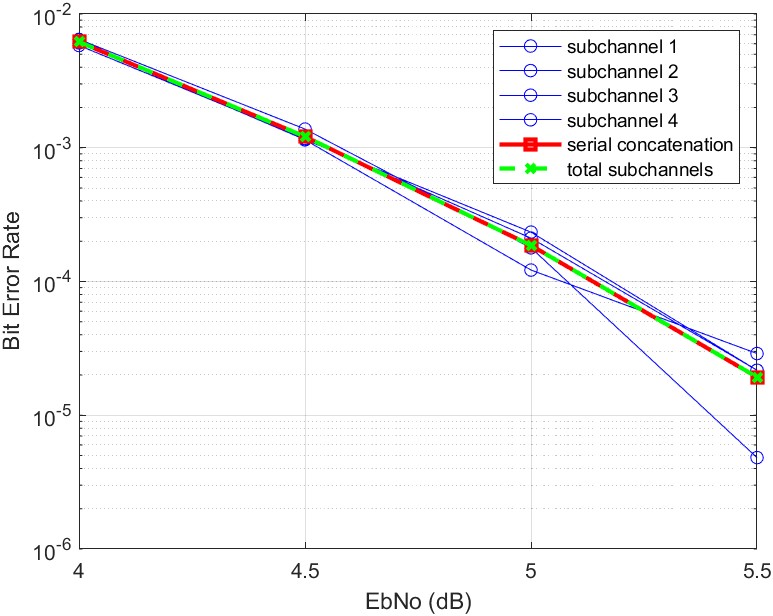}
    \vspace{-0.3cm}
    \caption{Second simulation.}
    \label{fig:sim2}
    \vspace{-0.4cm}
\end{figure}

\section{Conclusions and future work}\label{sec:conclusions}
\noindent In this work we proved that the concatenation of multiple received codewords does not affect the APPs of the information bits in a convolutional coding system,
      as long as the necessary termination conditions are met in the encoding process and the relative information is available to the receiver.
  Supported by simulation results, we showed the applicability of this result to realistic DAB frameworks; this opens the doors for future research on the performance
      improvements that can be achieved by considering the CIF as one codeword in the iterative decoders for DAB systems.
\vspace{-0.2cm}
\bibliographystyle{IEEEtran}
\bibliography{refs}

\begin{thebibliography}{1}
\providecommand{\url}[1]{#1}
\csname url@samestyle\endcsname
\providecommand{\newblock}{\relax}
\providecommand{\bibinfo}[2]{#2}
\providecommand{\BIBentrySTDinterwordspacing}{\spaceskip=0pt\relax}
\providecommand{\BIBentryALTinterwordstretchfactor}{4}
\providecommand{\BIBentryALTinterwordspacing}{\spaceskip=\fontdimen2\font plus
\BIBentryALTinterwordstretchfactor\fontdimen3\font minus \fontdimen4\font\relax}
\providecommand{\BIBforeignlanguage}[2]{{%
\expandafter\ifx\csname l@#1\endcsname\relax
\typeout{** WARNING: IEEEtran.bst: No hyphenation pattern has been}%
\typeout{** loaded for the language `#1'. Using the pattern for}%
\typeout{** the default language instead.}%
\else
\language=\csname l@#1\endcsname
\fi
#2}}
\providecommand{\BIBdecl}{\relax}
\BIBdecl

\bibitem{etsi300}
E.~ETSI, ``{300 401 V1. 4.1 (2006-06) Radio Broadcasting Systems},'' \emph{{Digital Audio Broadcasting (DAB) to Mobile, Portable and Fixed Receivers}}, 2006.

\bibitem{hagenauer1996iterative}
J.~Hagenauer, E.~Offer, and L.~Papke, ``{Iterative Decoding of Binary Block and Convolutional Codes},'' \emph{{IEEE Transactions on Information Theory}}, vol.~42, no.~2, pp. 429--445, 1996.

\bibitem{peleg2000iterative}
M.~Peleg, S.~Shamai, and S.~Galan, ``{Iterative Decoding for Coded Noncoherent MPSK Communications over Phase-Noisy AWGN Channel},'' \emph{{IEE Proceedings-Communications}}, vol. 147, no.~2, pp. 87--95, 2000.

\bibitem{van2012two}
W.~J. van Houtum, ``{Two-Dimensional Block-Based Reception for Differentially Encoded OFDM Systems: a Study on Improved Reception Techniques for Digital Audio Broadcasting Systems},'' \emph{Ph.D. Thesis}, 2012.

\bibitem{bahl1974optimal}
L.~Bahl, J.~Cocke, F.~Jelinek, and J.~Raviv, ``{Optimal Decoding of Linear Codes for Minimizing Symbol Error Rate (corresp.)},'' \emph{{IEEE Transactions on Information Theory}}, vol.~20, no.~2, pp. 284--287, 1974.

\bibitem{viterbi1967error}
A.~Viterbi, ``{Error Bounds for Convolutional Codes and an Asymptotically Optimum Decoding Algorithm},'' \emph{IEEE Transactions on Information Theory}, vol.~13, no.~2, pp. 260--269, 1967.

\end{thebibliography}

\end{document}